\documentclass[10pt,conference]{IEEEtran}
\IEEEoverridecommandlockouts
\usepackage[dvipdf]{graphicx,color}
\usepackage{amssymb}
\usepackage{amsmath}
\usepackage{stfloats}
\usepackage{amsfonts}
\usepackage{balance}
\usepackage{color}
\usepackage{algorithm}
\usepackage{algpseudocode}
\usepackage{color}
\usepackage{varwidth}
\usepackage{multicol}
\usepackage{subfigure}
\usepackage{xspace}
\usepackage{enumerate}
\usepackage{xcolor,cite,etoolbox}
\usepackage{bm}
\usepackage{amsthm}
\newtheorem{proposition}{Proposition}

\newtheorem{remark}{Remark}
\usepackage[font=footnotesize,labelfont=footnotesize]{caption}
\usepackage{ragged2e}

\def\BibTeX{{\rm B\kern-.05em{\sc i\kern-.025em b}\kern-.08em
    T\kern-.1667em\lower.7ex\hbox{E}\kern-.125emX}}

\usepackage[font=footnotesize,skip=6.4pt]{caption}
\begin{document}

\title{Impact of CSIR, SIC, and Hardware Impairments on the Ergodic Rate of Downlink RSMA}
\author{\IEEEauthorblockN{Farjam Karim\IEEEauthorrefmark{1},   Deepak Kumar\IEEEauthorrefmark{3}, Prathapasinghe Dharmawansa\IEEEauthorrefmark{1},\\
    Nurul Huda Mahmood\IEEEauthorrefmark{1}, Arthur Sousa de Sena\IEEEauthorrefmark{1}, and Matti Latva-aho\IEEEauthorrefmark{1}%
        \thanks{} 
	\thanks{This research was supported by the Research Council of Finland (former Academy of Finland) through the 6G Flagship  (Grant Number: 369116), Business Finland's 6GBridge-6CORE Project (Grant Number: 8410/31/2022) and the Riitta Ja Jorma J. Takasen säätiö grant.}}\\

	\IEEEauthorblockA{\IEEEauthorrefmark{1}Centre for Wireless Communications, University of Oulu, Finland. \\ 
    \IEEEauthorblockA {\IEEEauthorrefmark{3} Computer and Information Engineering, Khalifa University, Abu Dhabi, United Arab Emirates.}  
		Email: \{farjam.karim,\;Prathapasinghe.KaluwaDevage,\;nurulhuda.mahmood,\;arthur.sena,\;matti.latva-aho\}@oulu.fi,\\
        deepak.kumar@ku.ac.ae
		}}

	\maketitle
	
\begin{abstract}
This work investigates the ergodic rate performance analysis of rate-splitting multiple access (RSMA) in a downlink communication system under practical impairments. Closed-form expressions are derived for key performance metrics such as ergodic rate, energy efficiency, sum-rate, and Jain’s fairness index, capturing the joint effects of imperfect channel state information at the receiver (CSIR), imperfect successive interference cancellation (SIC), and hardware impairments. Numerical simulations validate the accuracy of the analytical expressions and reveal several insightful trends. At low transmit powers, imperfect CSIR is the dominant performance-limiting factor, followed by hardware impairments and imperfect SIC. However, as the transmit power increases, hardware impairments become the primary bottleneck, with the impact of imperfect CSIR gradually diminishing, and imperfect SIC becoming a more prominent bottleneck. Moreover, RSMA consistently outperforms non-orthogonal multiple access (NOMA) in terms of ergodic rate, fairness, and sum-rate, even under severe non-idealities. These findings underscore the importance of incorporating fairness as a core design objective alongside rate and energy efficiency, positioning RSMA as a robust and strong multiple access candidate for next-generation wireless networks.
\end{abstract}
 

\section{Introduction}
The sixth generation (6G) of wireless networks is expected to revolutionize the communication landscape by enabling a fully connected, intelligent, and immersive digital world~\cite{Matti_2019}. With ambitious goals such as supporting ultra-reliable low-latency communication, massive machine-type communication and extremely high data rates, 6G networks must accommodate a highly heterogeneous set of devices, services, and performance requirements~\cite{Nurul_eucnc_2020}. These demands impose significant challenges on the physical layer, especially in terms of spectrum utilization, interference management, and scalability.

To address the growing demands of future wireless networks, next-generation multiple access (NGMA) schemes must overcome the limitations of traditional orthogonal multiple access (OMA) techniques and offer greater flexibility in interference management and spectrum utilization. Among the prominent candidates, non-orthogonal multiple access (NOMA) and rate-splitting multiple access (RSMA) have emerged as promising solutions for 6G networks~\cite{Liu_yuan_jsac_22, Mao_com_soc_22}. Although both rely on superposition coding at the transmitter, their core operational principles differ significantly. In downlink NOMA, users are assigned distinct power levels based on their channel conditions, and multi-step successive interference cancellation (SIC) is applied at the receiver to decode signals in a hierarchical order. In contrast, downlink RSMA adopts a more versatile strategy by splitting each user’s message into a common part and a private part. The common parts from all users are combined and encoded into a single common stream, while the private parts are independently encoded into separate private streams. At the receiver, each user first decodes the common stream by treating all private streams as noise. Once the common stream is successfully decoded and removed, the user proceeds to decode their own private stream while treating the private streams of other users as noise. This mechanism enables RSMA to manage interference more flexibly by partially decoding interference through the common stream and partially treating the remaining interference as noise during private stream decoding. Moreover, RSMA decoding process consists of single step SIC~\cite{Farjam_WCNC_24}. Owing to these properties, RSMA is able to achieve significant performance gains with only a single layer of SIC, eliminating the need for complex multi-step decoding. This simplicity, combined with its flexibility in interference management, positions RSMA as a highly attractive multiple access strategy for 6G wireless networks.
 
Building on these advantages, RSMA has attracted growing research interest across various domains of wireless communications. For example, the authors in~\cite{Farjam_WCNC_24} have analyzed the system performance of a simultaneous wireless information and power transfer (SWIPT)-aided downlink RSMA system, considering imperfect channel state information (CSI) and imperfect SIC at the receiver. Expressions for the outage probability, throughput, and harvested energy are derived. In~\cite{wang_tcom_25}, the authors developed an optimization framework to maximize the ergodic sum-rate for a generalized rate-splitting system under imperfect CSI at the transmitter. In~\cite{jinping_lett_24}, the authors introduced the integration of RSMA into high-mobility wireless systems with stringent finite blocklength (FBL) and block error rate (BLER) requirements, and derived a closed-form expression for the ergodic sum-rate under imperfect CSI. The authors in~\cite{liu_qiu_tvt_25}, conducted a comprehensive performance analysis and energy efficiency (EE) optimization of an unmanned-aerial vehicle assisted dual-layer heterogeneous communication network employing RSMA.  Moreover, in~\cite{nguyen_tcom_23}, a dynamic framework was developed to jointly optimize rate control and power allocation, considering uncertainties in the surrounding environment and constraints related to covert communications. Lastly, in~\cite{Farjam_twc_25}, the authors derived analytical expressions for BLER and goodput, and both linear and non-linear energy harvesting in a SWIPT-enabled-RSMA-aided system, while accounting for imperfect CSI, imperfect SIC, and hardware impairments under FBL transmission.


Although RSMA has been extensively studied, most existing works assume ideal system conditions or, at best, account only for imperfect CSI. In practice, however, wireless systems are affected by a combination of impairments including imperfect CSI at the receiver (CSIR), imperfect SIC, and non-ideal hardware impairments, that are often overlooked. Furthermore, much of the prior literature focuses on outage probability or optimization-based formulations, which may not offer the same level of analytical insight into long-term system behavior as ergodic performance analysis.

To address these gaps, this work develops a comprehensive analytical framework for an RSMA-enabled downlink system operating under practical impairments. Closed-form expressions are derived for key performance metrics, including ergodic rate, EE, and sum-rate, while explicitly incorporating the effects of imperfect CSIR, imperfect SIC, and non-ideal  hardware impairments. Additionally, to evaluate fairness, we present an analytical expression for Jain’s fairness index (JFI), enabling a deeper exploration of the trade-off between efficiency and fairness in realistic RSMA deployments. The analytical results are validated through extensive numerical simulations. Finally, a detailed comparison with NOMA is provided, demonstrating that RSMA is significantly more robust to imperfect CSIR and illustrating how idealized assumptions such as perfect CSI, perfect SIC, and ideal hardware can lead to overly optimistic performance estimates, which may not reflect the realities of future 6G systems.

\noindent \textbf{Structure of the paper:} Section~\ref{Sys_main} presents the system model adopted in this study. In Section~\ref{Sec_Analysis}, closed-form expressions are derived for the ergodic rate, EE, and JFI under both perfect and imperfect conditions. Section IV provides a comprehensive discussion of the results. Finally, Section V concludes the paper by summarizing the key findings.

\noindent \textbf{Notations:} $\mathcal{CN}(0, \sigma^2_{(\cdot)})$ denotes a circularly symmetric complex Gaussian distribution with zero mean and variance $\sigma^2_{(\cdot)}$. The probability density function (PDF) of a random variable $X$ is expressed as $f_X(x)$, while $\Gamma(\cdot)$ and $\Gamma(\cdot, \cdot)$ represents the Gamma function and the upper incomplete gamma function, respectively. The parameters $m_{(\cdot)}$ and $\hat{\Omega}_{(\cdot)}$ correspond to the shape and scale parameters of the Nakagami-$m$ fading.
 

\begin{figure}[t]
    \centering
    \includegraphics[width=0.8\columnwidth]{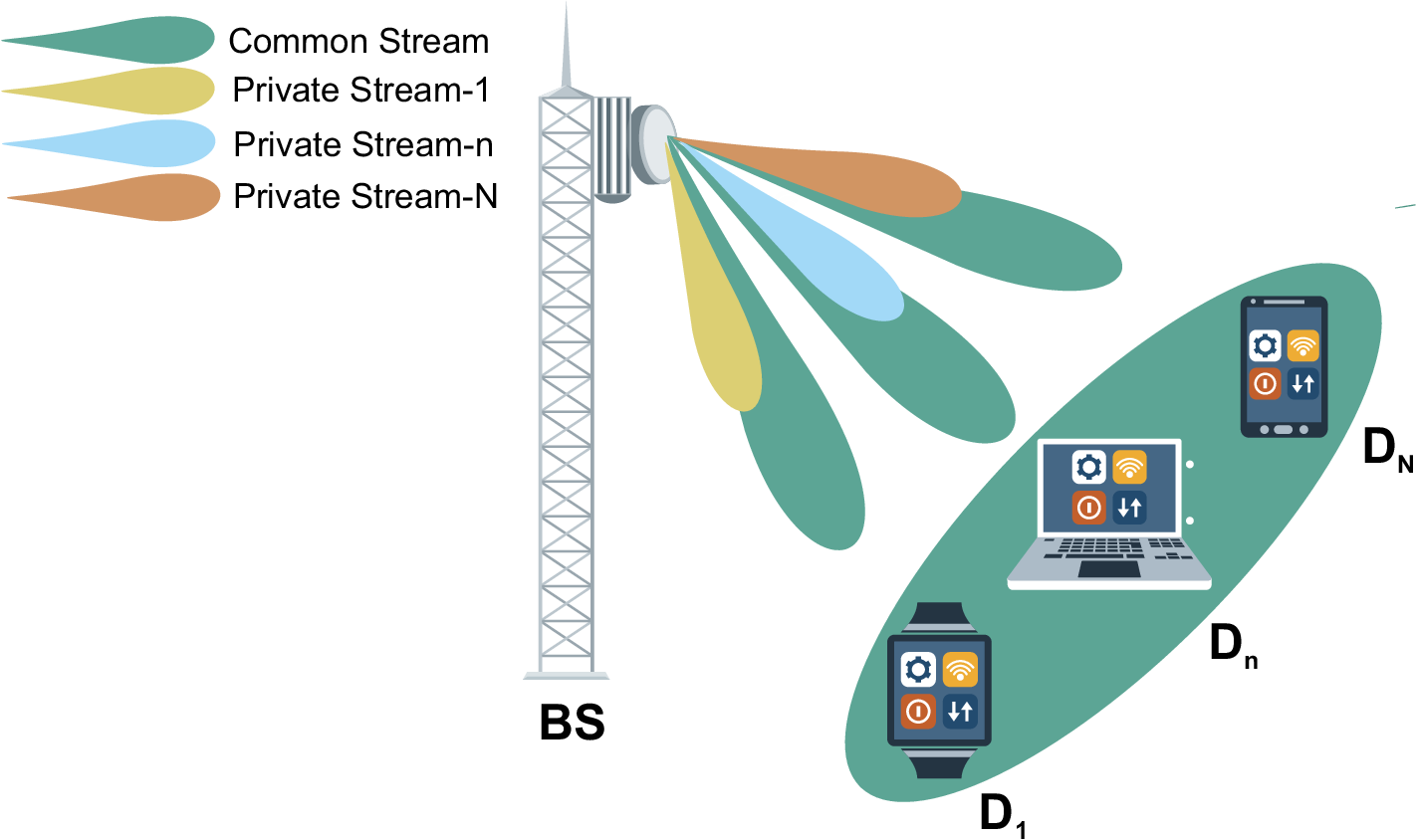}
    \caption{Illustration of a general RSMA-aided Downlink System with $N$ users sharing a single common stream for all users and $N$ private streams for each individual user.}
    \label{fig:system_diag}
\end{figure}
\section{System Model}
\label{Sys_main}
We consider a downlink RSMA-aided system, as illustrated in Fig.~\ref{fig:system_diag}, where a base station (BS) equipped with a single transmit antenna serves $N$ single-antenna users. To communicate with the users, the BS employs an RSMA scheme in which the common parts of all users’ messages are superimposed into a single common stream, denoted by $x_c$, along with $N$ private streams $x_n$, one for each user.
The transmitted signal from the BS is expressed as
$\left(x = \sqrt{P\beta_c}x_c + \sum_{n=1}^{N} \sqrt{P\beta_n}x_n\right)$,
where $\beta_c$ and $\beta_n$ are the power allocation coefficients for the common and $n_\text{th}$ private streams, respectively, and $P$ denotes the BS transmit power.
The channel gain between the BS and the $n$-th user is modeled as
$g_n = \hat{g}_n + g_{ne}$,
where $\hat{g}_n$ is the estimated channel gain, and $g_{ne}$ represents the channel estimation error (CEE)\cite{Deepak_SWIPT_SystemJ_23Mar}. All estimated channel gains are assumed to follow Nakagami-$m$ fading; hence, the squared channel amplitudes follow a Gamma distribution~\cite{Karim_WCL_25}. According to\cite{Deepak_SWIPT_SystemJ_23Mar}, the variance of the estimation error is given by
$\Omega_{gne} = \Omega_n / (1 + \rho \xi_n \Omega_n)$,
where $\Omega_n$ is the total variance of $g_n$, $\rho$ is the signal-to-noise ratio (SNR), and $\xi_n$ denotes the channel estimation quality parameter for the $n_\text{th}$ user. Accordingly, the variance of the estimated channel gain is
$\hat{\Omega}_n = \Omega_n - \Omega_{gne}$.
Thus, the signal received at the $n_{\text{th}}$ downlink user can be given as 
   \begin{align}\label{received signal}
     y_n=\Bigg[\left(\hat{g}_n +{g}_{ne}\right)\left(x
     +\eta_t+\eta_{rn}\right)\Bigg]+w_n,
 \end{align}
 where $\eta_t\sim \mathcal{CN}\left(0, P\kappa^2_{t}\right)$ and $\eta_{rn}\sim \mathcal{CN}\left(0, P\kappa^2_{rn}\right)$ denote the distortion noises, with $\kappa^2_{t}$ and $\kappa^2_{rn}$ being the level of
hardware impairments and $\omega_n\sim\mathcal{CN}(0,\sigma^2_n)$ is the additive white Gaussian noise (AWGN). To successfully retrieve its intended message, the $n_{\text{th}}$ user first decodes the common stream from the received signal. The signal-to-interference-plus-noise ratio (SINR) for decoding the common stream is given by
\begin{align} \label{SINR_comm_down}
    \gamma_{c,n} &= \frac{ \rho\beta_c|\hat{g}_n|^2}{\begin{aligned}   &1+\rho\left(\sum\limits^{N}_{n=1}\beta_n+\kappa^2_{rn}+\kappa^2_t\right)\left|\hat{g}_n\right|^2 \\
&+\rho\Omega_{g_{ne}}\left(\left(\beta_{c}+\sum^{N}_{n=1}\beta_{n}\right)+\kappa^2_{rn}+\kappa^2_t\right)
    \end{aligned}},
\end{align}
where $\rho=P/\sigma^2_{n}$. After successfully decoding the common stream, the $n_{\text{th}}$ user  attempts to decode its own private stream with a SINR of  
\begin{align} \label{SINRpvt}
    \gamma_{p,n} &= \frac{ \rho\beta_k|\hat{g}_n|^2}{\begin{aligned}   &1+\rho\left(\sum\limits_{\substack{j=1\\ j\neq{n}}}^{N}\beta_j+{\varphi_n \beta_c}+\kappa^2_{rn}+\kappa^2_t\right)\left|\hat{g}_n\right|^2 \\
&+\rho\Omega_{g_{ne}}\left(\left(\varphi_n\beta_{c}+\sum^{N}_{n=1}\beta_{n}\right)+\kappa^2_{rn}+\kappa^2_t\right)
    \end{aligned}},
\end{align}
where $\varphi_n\in(0,1)$ denotes imperfect SIC factor with $\varphi=0$ being perfect SIC.

\begin{figure*}
    \begin{align}\label{final_comm_ER}
        \zeta_1 \simeq \min_{n\in N}\Bigg[\frac{D_1^{m_n}C_1\bigl[A_2(A_1+\frac{C_1}{2})^{-1}\bigl]^{{m_n}}\Gamma(m_n+1)\exp(D_1A_2(A_1+\frac{C_1}{2})^{-1})\Gamma\left(-m_n, \left(D_1A_2(A_1+\frac{C_1}{2})^{-1}\right)\right)}{\ln(2)\Gamma(m_n)(A_1+\frac{C_1}{2})}\Bigg].
    \end{align}\hrulefill
    \begin{align}\label{pvt_ER}
        \zeta_2 \simeq \frac{D_1^{m_n}C_2\bigl[B_2(B_1+\frac{C_2}{2})^{-1}\bigl]^{{m_n}}\Gamma(m_n+1)\exp(D_1B_2(B_1+\frac{C_2}{2})^{-1})\Gamma\left(-m_n, \left(D_1B_2(B_1+\frac{C_2}{2})^{-1}\right)\right)}{\ln(2)\Gamma(m_n)(B_1+\frac{C_2}{2})}.
    \end{align}\hrulefill
   
\end{figure*}
    \section{Performance Analysis}\label{Sec_Analysis}
  In this section, we derive closed-form  expressions for the ergodic rate and JFI under both ideal and non-ideal conditions. Based on the derived ergodic rate expressions, we further analyze the EE and sum-rate of the system.
 As previously discussed, all users' common messages are jointly encoded into a single common stream, which must be decoded by every user. Each private message is independently encoded into a private stream, intended to be decoded only by its corresponding user.
Accordingly, the ergodic rate for the  $n_{\text{th}}$ user is given by
   \begin{align}\label{first}
       \mathcal{R}_{n}=\! \left(\min_{n\in N}\mathbb{E}\big[ 
\log_{2} \left(1\!+\gamma_{c,n}\right) \big]\right) \!+\! 
\mathbb{E}\big[ 
\log_{2} \left(1+\!\gamma_{p,n}\right) \Big],
   \end{align}
   where the $\min$ operator ensures that the common stream $x_c$ is decoded by all users.
 \begin{proposition}
     The approximated ergodic rate at the $n_{\text{th}}$ user considering imperfect CSIR, SIC and hardware impairments can be expressed as
     \begin{align}\label{ER_main}
         \mathcal{R}_n\simeq \zeta_1+\zeta_2,
     \end{align}
     where $\zeta_1$ and $\zeta_2$ are given in \eqref{final_comm_ER} and \eqref{pvt_ER}, respectively, on the top of this page, such that $C_1 = \rho\beta_c$, $A_1 =\rho\left(\sum\limits^{N}_{n=1}\beta_n+\kappa^2_{rn}+\kappa^2_t\right)$, $A_2=\rho\Omega_{g_{ne}}\left(\left(\beta_{c}+\sum^{N}_{n=1}\beta_{n}\right)+\kappa^2_{rn}+\kappa^2_t\right)+1$, $D_1=\frac{m_n}{\hat{\Omega}_n}$, $C_2=\rho\beta_k$, $B_1=\rho\left(\sum\limits_{\substack{j=1\\ j\neq{n}}}^{N}\beta_j+{\varphi_n \beta_c}+\kappa^2_{rn}+\kappa^2_t\right)$ and $B_2=\rho\Omega_{g_{ne}}\left(\left(\varphi_n\beta_{c}+\sum^{N}_{n=1}\beta_{n}\right)+\kappa^2_{rn}+\kappa^2_t\right)+1$.
 \end{proposition}
      \begin{proof}
  In order to derive $n_{\text{th}}$ user ergodic rate, we need to solve for \eqref{SINR_comm_down}  and then \eqref{SINRpvt}. Therefore, we write \eqref{first} as
   \begin{align}\label{down_proof}
   \mathcal{R}_{n}= \underbrace{\min_{n\in N}\mathbb{E}\big[ 
\log_{2} \left(1+\gamma_{c,n}\right)}_{\zeta_1} \big] + 
\underbrace{\mathbb{E}\big[ 
\log_{2} \left(1+\gamma_{p,n}\right)}_{\zeta_2} \Big].
\end{align}
We first evaluate the ergodic rate expression for $\zeta_1$. Let $X=|\hat{g}_n|^2$ and  by using the approximations 
 $\log_2(1+x)\simeq \frac{2x}{(2+x)\ln(2)}$, \cite{Topsoe}, which is tight for $x \leq 1$. This condition holds in our setting, as the signal and interference powers are of comparable magnitude due to the RSMA procedure. Thus, we can express $\zeta_1$ as
\begin{align}\label{down_step_1}
   \zeta_1\simeq \min_{n\in N}\int\limits^{\infty}_{0}\frac{2\left(\frac{C_1x}{A_1x+A_2}\right)}{\Bigg[2+\left(\frac{C_1x}{A_1x+A_2}\right)\Bigg] \ln(2)}f_{X}(x)dx,
\end{align}
where $C_1 = \rho\beta_c$, $A_1 =\rho\left(\sum\limits^{N}_{n=1}\beta_n+\kappa^2_{rn}+\kappa^2_t\right)$, and $A_2=\rho\Omega_{g_{ne}}\left(\left(\beta_{c}+\sum^{N}_{n=1}\beta_{n}\right)+\kappa^2_{rn}+\kappa^2_t\right)+1$. As $X$ follows a Gamma distribution and performing some mathematical simplifications, we can express \eqref{down_step_1} as 
\begin{align}\label{down_Step_2}
        \zeta_1\simeq \min_{n\in N}\Bigg[\frac{C_1 D_1^{m_n}}{\ln(2)\Gamma(m_n)}\int\limits^{\infty}_{0}\frac{x^{m_n}\exp(-D_1x)}{A_1x+A_2+\frac{C_1x}{2}}dx\Bigg],
         \end{align}
where $D_1=\frac{m_n}{\hat{\Omega}_n}$. By taking out the term $(A_1+\frac{C_1}{2})$ as common from the denominator of \eqref{down_Step_2}, we obtain
\begin{align}\label{down_Step_3}
     \zeta_1\!\simeq\!\min_{n\in N} \!\Bigg[\frac{C_1 D_1^{m_n}}{\ln(2)\Gamma(m_n) (A_1\!+\!\frac{C_1}{2})}\!\!\int\limits^{\infty}_{0}\!\!\!\frac{x^{m_n}\exp(-D_1x)}{\bigl[A_2(A_1+\frac{C_1}{2})^{-1}\!+\!x\bigl]}dx\Bigg].
\end{align}
Using~\cite[Eq. $3.383.10$]{2015249}, we solve the integral and obtain the expression for $\zeta_1$ in \eqref{final_comm_ER}. Similarly, we solve for $\zeta_{2}$ and obtain the ergodic rate expression in \eqref{pvt_ER}. Finally using \eqref{final_comm_ER} and \eqref{pvt_ER} expressions, we obtain \eqref{ER_main}.
 \end{proof}
 \begin{remark}
     By substituting $\xi_n\to\infty$, $\varphi_n=0$, and $\kappa^{2}_{rn}=\kappa^{2}_{t}=0$, the ergodic rate for the ideal case can be obtained.
 \end{remark}
 \begin{remark}
  The system-level EE, accounting for imperfect CSIR, SIC, and hardware impairments, is expressed as
$\varphi_{\text{EE}} = \frac{\sum_{n=1}^{N} \mathcal{R}_n}{P+P_{c}},$
where \(\sum_{n=1}^{N} \mathcal{R}_n\) is the system sum rate, $P_c$ is the BS circuit power consumption and $P$ is the BS transmit power. Moreover, leveraging Remark.~1, the system-level EE and sum-rate can be obtained for the ideal case.
 \end{remark}
 \begin{proposition}
     The JFI considering imperfect SIC and hardware impairments can be evaluated as
     \begin{align}\label{JFI}
         \mathbb{J} = \frac{\left(\sum^{N}_{n=1}\mathcal{R}_n\right)^2}{N\sum^{N}_{n=1}\mathcal{R}^2_n},
     \end{align}
     where $R_n$ is given in \eqref{ER_main}. By using Remark.~1, the JFI for the ideal scenario can also be obtained.
 \end{proposition}
\section{Numerical Results}\label{Sec:Discussion}
 This section presents Monte Carlo simulation results to validate the derived closed-form expressions.
The wireless link between the BS and the $n_{\text{th}}$ user is characterized using a distance-based path-loss model, given by $\delta/\mathcal{D}_n^{\tau_{n}}$, where  $\delta = 1$~meter is the reference path-loss parameter, $\mathcal{D}_n$ is the distance and $\tau_{n}$ is the path-loss exponent. Note that, all the derived expressions are valid for $N$ devices. However, for clarity, results are plotted only for two downlink users, $D_1$ and $D_2$. 
  \begin{table}[h!]	\renewcommand{\arraystretch}{1.0}
		\centering
		\caption{ Simulation Parameters.}
		\label{t3}
			\resizebox{\columnwidth}{!}{\begin{tabular}{|l|l|l|l|l|l|}
			\hline
			Parameter         & Value         & Parameter & Value  & Parameter & Value  \\ \hline
			$m_1  $   &     $4 $   	&  $m_{2}$  &     $4$ & $\sigma^2_n$   &     $-100 $ dB \\ \hline
			
			$\beta_{c}$       & $0.6$ &	   $\beta_{1}$       & $0.25 $ &  $\beta_2 $  &     $0.15$  	 \\ \hline
			
		 BS-$D_1$      & $135 $~m 	&  BS-$D_2$     & $120 $~m 	&    $\tau_n$    &   $3.6$   	 \\ \hline




             
		
		\end{tabular}}
	\end{table}

        \begin{figure*}[!t]
  \begin{minipage}[b]{0.495\textwidth}
    \centering
    \includegraphics[width=\textwidth]{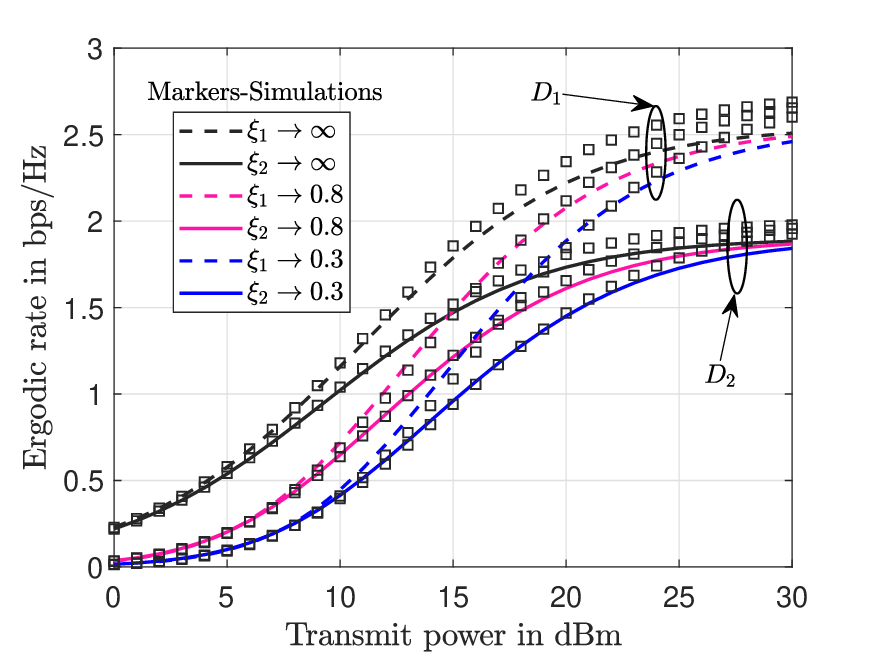}
    \caption{Impact of $\xi_n$ on the user's ergodic rate.}
    \label{CEE}
  \end{minipage}
  \begin{minipage}[b]{0.495\textwidth}
    \centering
    \includegraphics[width=\textwidth]{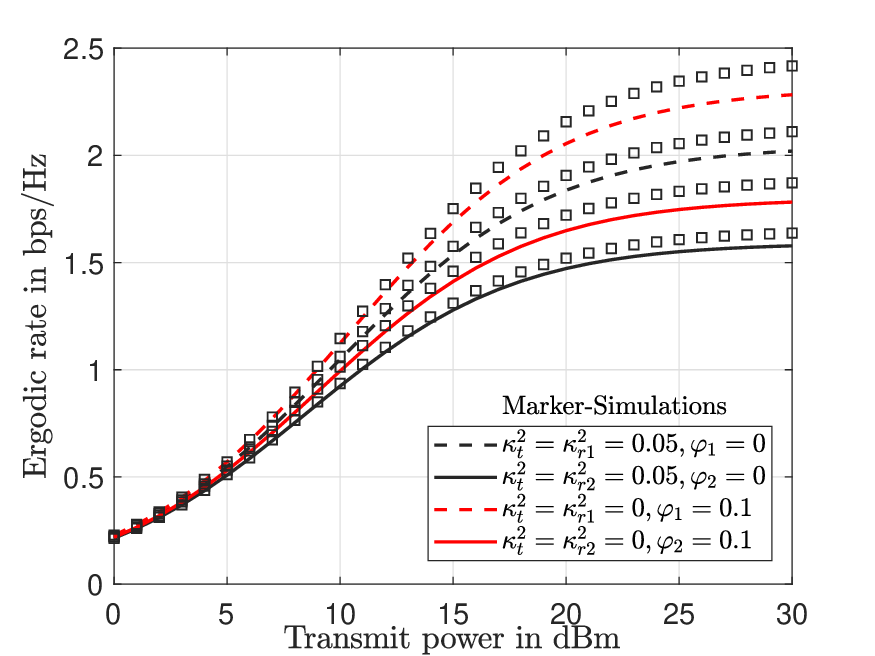}
    \caption{Effect of imperfect SIC and hardware impairments.}
    \label{SIC_HI}
  \end{minipage}
 \end{figure*}
 
Fig.~\ref{CEE} and Fig.~\ref{SIC_HI} illustrate the ergodic rate $\mathcal{R}_n$, measured in bits per second per Hertz (bps/Hz), for users $D_1$ and $D_2$ as a function of the BS transmit power $P$. In Fig.~\ref{CEE}, $\mathcal{R}_1$ and $\mathcal{R}_2$ are plotted for different values of the CEE parameter $\xi_n$, assuming perfect SIC and ideal hardware, thereby isolating the effect of imperfect CSIR on system performance. Conversely, Fig.~\ref{SIC_HI} analyzes the impact of imperfect SIC and hardware impairments, assuming perfect CSIR to enable a fair comparison.
In both figures, the analytical curves closely match the simulation results, validating the accuracy of the derived expressions and confirming the reliability of the associated corollaries under both ideal and non-ideal conditions. As expected, $\mathcal{R}_n$ increases with increasing $P$, but saturates beyond approximately $P \geq 25$~dBm. This saturation is primarily due to  interference from other users' streams.
A comparative analysis of the two figures offers further insights. Imperfect CSIR has a more pronounced impact at low transmit powers; however, its effect diminishes as $P$ increases, consistent with the CEE model, where channel estimation improves with SNR. This trend is evident in Fig.~\ref{CEE}, where the curves for different $\xi_n$ values $\in$ $\{0.3, 0.8, \infty\}$ begin to converge and saturate around $P \approx 28$~dBm, indicating reduced sensitivity to channel estimation accuracy at higher power levels.
In contrast, the degradation due to imperfect SIC and hardware impairments persists and even becomes more pronounced at higher transmit powers. For instance, when $\xi_n = 0.3$ or $0.8$ in Fig.~\ref{CEE}, both $\mathcal{R}_1$ and $\mathcal{R}_2$ remain well below $0.5$ bps/Hz at $P\approx 5$~dBm. However, for the same  transmit power in Fig.~\ref{SIC_HI}, the rates approach approximately above $0.5$ bps/Hz, highlighting the relatively milder impact of imperfect SIC and hardware at low power levels. At higher power levels ($P \approx 28$~dBm), $\mathcal{R}_1$ and $\mathcal{R}_2$ in Fig.~\ref{CEE} reach about $2.5$ and $1.75$ bps/Hz, respectively, whereas in Fig.~\ref{SIC_HI}, the rates are lower, demonstrating the increasing impact of SIC and hardware imperfections.
Moreover, between the two impairments, hardware impairments have a more severe effect on $\mathcal{R}_n$ than imperfect SIC, even when both impairment factors measurement are within the same range $(0,1)$. This behavior stems from two key reasons: (i) hardware impairments affect the entire transmission chain starting from the common stream decoding, whereas the impact of imperfect SIC begins only after the common stream is decoded, and manifests as residual interference; (ii) the distortion from hardware impairments scales directly with both transmit power and channel gain, while the effect of imperfect SIC is first attenuated by the power allocated to the common stream, thereby resulting in comparatively less degradation.

 \begin{figure*}
    \begin{minipage}[b]{0.33\textwidth}
    \centering
    \includegraphics[width=\textwidth]{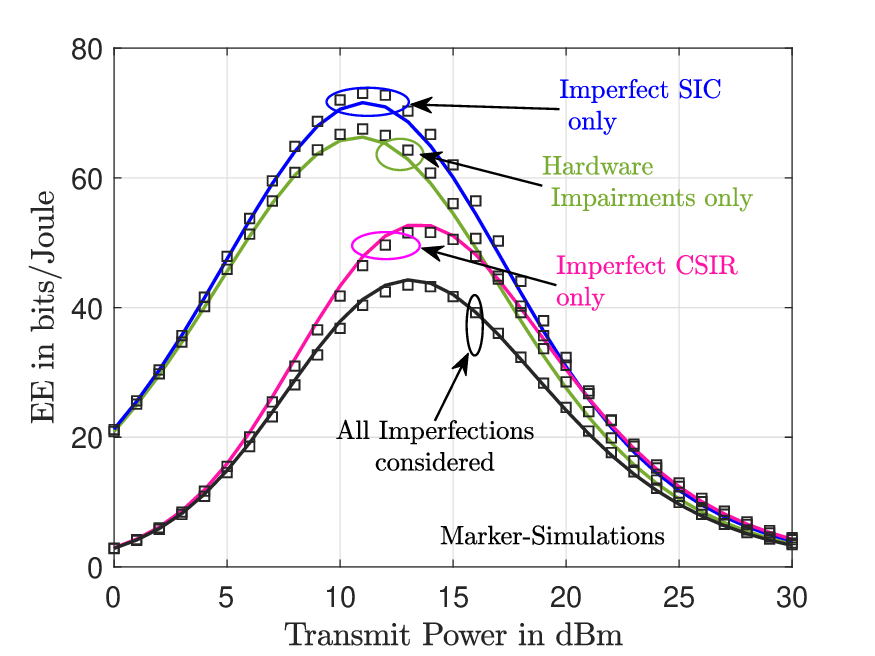}
     \caption{System's Energy Efficiency.}
     \label{EE}
     \end{minipage}
     \begin{minipage}[b]{0.33\textwidth}
    \centering
    \includegraphics[width=\textwidth]{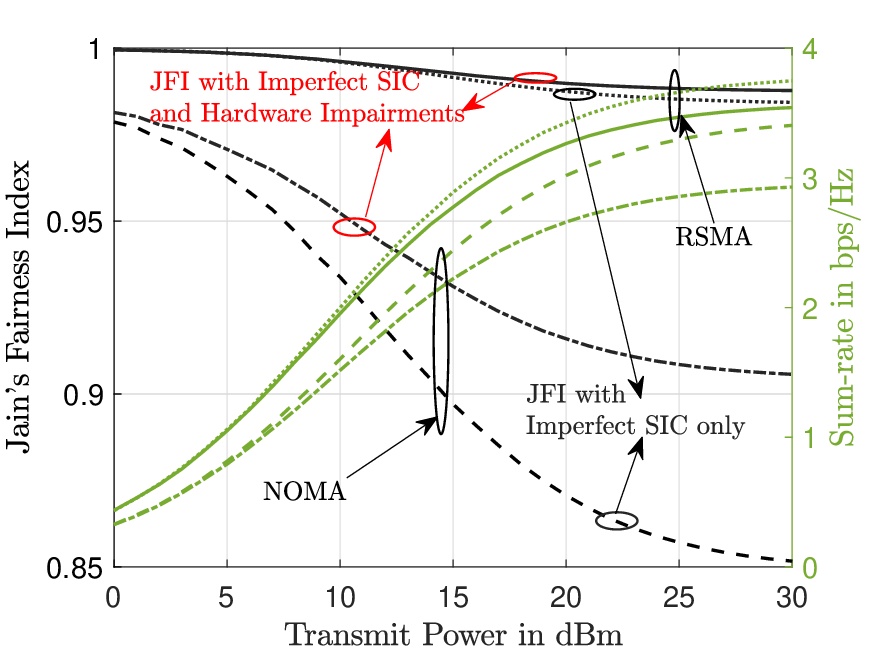}
    \caption{Fairness between RSMA and NOMA.}
    \label{JFI}
  \end{minipage}
  \begin{minipage}[b]{0.33\textwidth}
    \centering
    \includegraphics[width=\textwidth]{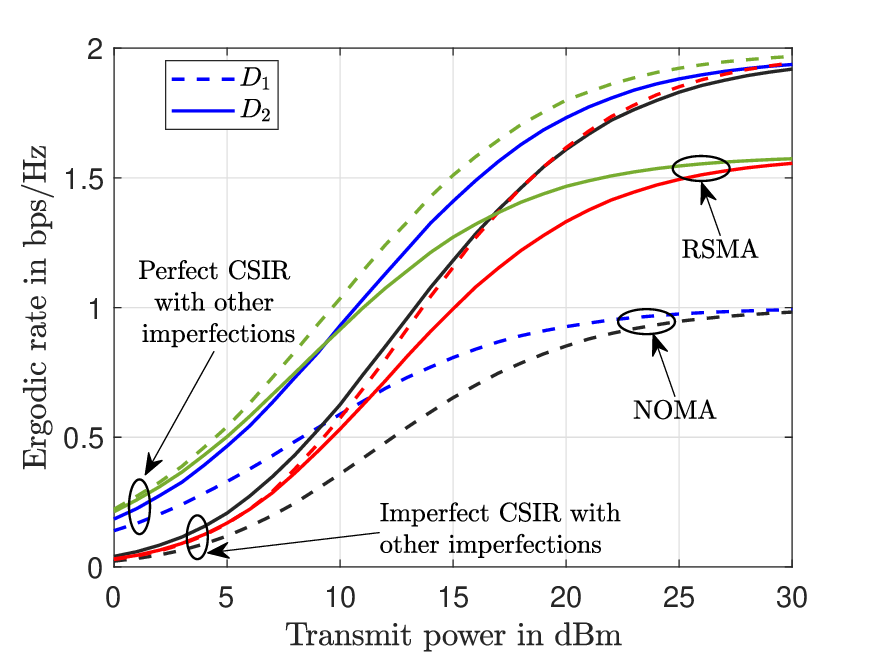}
    \caption{Comparison between RSMA and NOMA.}
    \label{Compare}
  \end{minipage}
 \end{figure*}

Fig.~\ref{EE} presents the system EE as a function of the BS transmit power $P$, considering the joint effects of imperfect CSIR, imperfect SIC, and hardware impairments. The circuit power consumption is fixed at $0.02$~Watts,  $\xi_n\to 0.7$, $\varphi_n=0.1$ and $ \kappa^2_t = \kappa^2_{rn} = 0.05 $. As evident from the figure, the analytical curves closely align with the simulation results, validating the accuracy of the analytical expression provided in Remark~2. As expected, EE initially increases with $P$, since the ergodic rate improves rapidly while the total power consumption grows more slowly. This is typical of the noise-limited regime, where an increase in transmit power directly enhances data rate with relatively less power consumption. However, beyond a certain transmit power threshold, EE peaks and begins to decline. This is because the ergodic rate increases logarithmically with $P$ due to the Shannon capacity formula, while the total power consumption, comprising $P$ and the fixed $P_c$, continues to rise linearly. As a result, the energy cost per bit increases, and EE begins to degrade, indicating diminishing returns. 
Further insights emerge when comparing the influence of different system imperfections. Imperfect CSIR has the most pronounced negative impact on EE at low transmit power levels, particularly up to around $P \approx 14$~dBm. At very low $P$ (e.g., $P \approx 2$~dBm), imperfect SIC and hardware impairments have a similar impact on EE. However, as $P$ increases, the degradation caused by hardware impairments becomes more significant than that caused by imperfect SIC. Additionally, at around $P \approx 16$~dBm, the EE under imperfect CSIR surpasses that under hardware impairments, and by $P \approx 21$~dBm, it is roughly equal to the EE observed with imperfect SIC. These observations reinforce earlier results from Fig.~\ref{CEE} and Fig.~\ref{SIC_HI}, emphasizing that system imperfections, especially hardware impairments can substantially influence system performance at moderate-to-high transmit powers. Therefore, for energy-efficient system design, it is critical to account for these impairments, particularly in practical 6G and beyond communication scenarios, where high data rates and power efficiency are both essential.

Fig.~\ref{JFI} presents the JFI (left $y$-axis) and the sum-rate (right $y$-axis) for both RSMA and NOMA schemes, plotted against varying transmit power~$P$. The results are shown under both ideal and non-ideal hardware impairments, with imperfect SIC and perfect CSIR assumed throughout. The impairment parameters are set to $\varphi_n = 0.1$ and $\kappa_t^2 = \kappa_{rn}^2 = 0.05$. For NOMA, power allocation coefficients for users $D_1$ and $D_2$ are set to 0.55 and 0.45, respectively.
Several key observations can be drawn from the JFI curves. First, RSMA consistently demonstrates superior fairness compared to NOMA across both impairment scenarios, highlighting its robustness. Second, and perhaps counterintuitively, the presence of system imperfections improves fairness. For instance, under NOMA with only SIC imperfections, the JFI begins around $0.97$ and gradually decreases with increasing~$P$. In contrast, when both SIC imperfections and hardware impairments are present, the JFI starts around $0.98$ and remains comparatively higher across all power levels. A similar trend is observed for RSMA, where the combination of imperfect SIC and hardware impairments results in higher fairness than the case with ideal hardware and imperfect SIC.
This seemingly paradoxical behavior can be attributed to the equalizing effect of impairments. Specifically, imperfect SIC and hardware distortions disproportionately degrade the performance of stronger users, thereby narrowing the rate gap between users. As a result, ergodic rates become more balanced, and fairness as measured by JFI  improves. In this context, system imperfections unintentionally act as fairness enhancers by limiting user dominance.
Meanwhile, the sum-rate trends (right $y$-axis) behave as expected: both RSMA and NOMA exhibit increasing rates with higher transmit power~$P$. However, scenarios with only imperfect SIC consistently yield higher sum-rates than those also affected by hardware impairments. Among the two schemes, RSMA maintains a consistent sum-rate advantage over NOMA across all power levels, thanks to its flexible message-splitting architecture and more effective interference management.
When JFI and sum-rate curves are considered jointly, an important design trade-off emerges: enhanced fairness under imperfections often comes at the cost of reduced sum-rate. Although imperfections degrade spectral efficiency, they simultaneously reduce the disparity between user rates, thereby improving fairness. This highlights a broader insight: fairness can be promoted not only through optimal resource allocation but also by constraining the performance of dominant users. Such constraints, however, must be carefully tuned to avoid disproportionately penalizing stronger users.
This trade-off has significant implications for the design of 6G and beyond wireless systems. Fairness should be treated as a primary performance metric not merely a by-product of ergodic rate maximization. A balanced approach is essential, depending on whether the application prioritizes peak data rates or consistent service quality across all users.

Fig.~\ref{Compare} illustrates the individual user ergodic rates for RSMA and NOMA under both perfect and imperfect CSIR, in the presence of non-ideal hardware impairments and imperfect SIC. For RSMA, the CSIR quality factor is set to $\xi_n \to 0.7$, while for NOMA, it is more favorable at $\xi_n \to 0.9$. Both schemes assume identical hardware impairment levels, with $\kappa^2_t = \kappa^2_{rn} = 0.05$ and SIC imperfection $\varphi_n = 0.1$.
While Fig.~\ref{JFI} already established RSMA's sum-rate advantage, this figure highlights that the benefit also extends to the individual user level. Notably, even under more challenging conditions for RSMA (i.e., lower CSIR quality), it still outperforms NOMA for $P \geq 19$~dBm. Below this power level ($P \leq 19$~dBm), RSMA remains close to NOMA, indicating its robustness to CSIR imperfections.
When perfect CSIR is assumed, and only hardware impairments and SIC imperfections are present, RSMA consistently delivers higher ergodic rates than NOMA. This performance gain is primarily due to RSMA’s flexible message-splitting architecture, where the inclusion of a common stream enables efficient interference management and more balanced resource allocation.
An interesting observation is the user-wise rate disparity between the two schemes. In NOMA, user $D_2$ consistently achieves higher rates than $D_1$ due to the decoding order imposed by the SIC process, which inherently favors the far user. In contrast, RSMA users do not decode others’ private messages, and thus, $D_1$ often benefits more depending on power allocation.
Lastly, it's important to note that in RSMA, interference from other users always remains while decoding private streams, even under ideal conditions. In NOMA, however, the last-decoded user ideally sees no interference if perfect SIC or ideal hardware impairments are considered, allowing its rate to grow more rapidly with increasing transmit power. This highlights the importance of incorporating practical imperfections when evaluating and comparing multiple access strategies for future networks.

\section{Conclusion}
In this work, we presented a comprehensive ergodic performance analysis of an RSMA-aided downlink system under practical impairments, including imperfect CSIR, non-ideal SIC, and hardware imperfections conditions often overlooked in prior works. Closed-form expressions were derived for key performance metrics such as ergodic rate, EE, sum-rate, and JFI, providing valuable system-level insights.
Our results revealed several key findings. First, the impact of imperfections is highly nuanced: while imperfect CSIR may degrade performance at low transmit powers, its negative effect diminishes at higher power levels. second, hardware impairments become increasingly dominant at high transmit powers, imposing significant limitations on both ergodic rate and EE. Third, RSMA consistently outperforms NOMA in both ergodic rate and fairness with imperfections even under worse conditions. Moreover, imperfections up to a certain extent can paradoxically enhance fairness. 
A clear trade-off was also observed between sum-rate and fairness under realistic conditions, reinforcing the need to treat fairness not as a secondary metric, but as a fundamental design criterion in future 6G systems. Notably, despite the relatively simple analytical setup, the insights obtained were both rich and impactful. This reflects the principle of Occam’s Razor, where simpler models, when properly formulated, can reveal deep and reliable understanding of complex nuances within the system. In this context, RSMA’s ability to sustain high spectral efficiency while significantly improving fairness under realistic conditions strengthens its position as a strong candidate for balanced performance in next-generation wireless networks.
 \bibliographystyle{IEEEtran_renamed}
	\bibliography{referencing}
\end{document}